\documentclass[letterpaper,11pt]{article}
\usepackage[letterpaper,margin=1in]{geometry}
\usepackage{hyperref}
\usepackage{amsmath}
\usepackage{amsthm}
\usepackage{latexsym}
\usepackage{graphicx}
\usepackage[utf8]{inputenc}
\usepackage[linesnumbered, ruled, vlined]{algorithm2e}
\usepackage{subfig}

\theoremstyle{plain}
\newtheorem{theorem}{Theorem}
\newtheorem{lemma}[theorem]{Lemma}

\graphicspath{{figures/}}

\newcommand{\fastdel}{\textsc{FastDelivery}\xspace}
\newcommand{\fastldel}{\textsc{FastLineDelivery}\xspace}
\newcommand{\dist}{\mathit{dist}}
\newcommand{\final}{\mathit{final}}
\newcommand{\APSP}{\mathrm{APSP}}
\newcommand{\ie}{\textit{i.e.}}
\newcommand{\eg}{e.\,g.\xspace}
\newcommand{\uset}[1]{\ifmmode\left\{\,#1\,\right\}\else\{\,#1\,\}\fi}

\usepackage{tkz-graph}
\usetikzlibrary{arrows}
\usetikzlibrary{arrows.meta}
\usetikzlibrary{shapes}
\usetikzlibrary{shapes.misc}
\usetikzlibrary{decorations.pathmorphing}
\usetikzlibrary{decorations.markings}

\tikzstyle{vertice} = [circle, text = black, fill = white, minimum size = 12pt, inner sep = 1pt, line width = .5pt, draw = black, font=\scriptsize]
\tikzstyle{root} = [vertice, text = white, fill = black]
\tikzstyle{edge} = [draw, thick, -]
\tikzstyle{boldEdge} = [edge, line width=3pt]
\tikzstyle{weight} = [font=\small]

\begin{document}
            
\title{An Efficient Algorithm for the Fast Delivery Problem}

\author{Iago A. Carvalho\thanks{Department of Computer Science, Universidade Federal de Minas Gerais, Brazil.
E-mail: iagoac@dcc.ufmg.br, URL: https://iagoac.github.io}
\and Thomas Erlebach\thanks{Department of Informatics, University of Leicester, England, UK.
E-mail: t.erlebach@leicester.ac.uk}
\and Kleitos Papadopoulos\thanks{Department of Informatics, University of Leicester, England, UK.
E-mail: kleitospa@gmail.com}
}

\maketitle

\begin{abstract}
We study a problem where $k$ autonomous mobile agents are initially located on distinct nodes of a weighted graph (with $n$ nodes and $m$ edges). Each autonomous mobile agent has a predefined velocity and is only allowed to move along the edges of the graph. We are interested in delivering a package, initially positioned in a source node $s$, to a destination node $y$. The delivery is achieved by the collective effort of the autonomous mobile agents, which can carry and exchange the package among them. The objective is to compute a delivery schedule that minimizes the delivery time of the package. In this paper, we propose an $\mathcal{O}(kn\log n+km)$ time algorithm for this problem. This improves the previous state-of-the-art $\mathcal{O}(k^2 m + k n^2 + \APSP)$ time algorithm for this problem, where $\APSP$ stands for the running-time of an algorithm for the All-Pairs Shortest Paths problem.
\end{abstract}

\section{Introduction} \label{sec:intro}

Enterprises, such as DHL, UPS, Swiss Post, and Amazon, are now delivering goods and packages to their clients using \emph{autonomous drones}~\cite{Bamburry2015,Regev2018}. Those drones depart from a base (which can be static, such as a warehouse~\cite{Hong2017}, or mobile, such as a truck or a van~\cite{Murray2015}) and deliver the package into their clients' houses or in the street. However, packages are not  delivered to a client that is too far from the drone's base due to the energy limitations of such autonomous aerial vehicles.

In the literature, we find some proposals for delivering packages over a longer distance. One of them, proposed by Hong, Kuby, and Murray~\cite{Hong2017}, is to install recharging bases in several spots, which allows a drone to stop, recharge, and continue its path. However, this strategy may result in a delayed delivery, because drones may stop several times to recharge during a single delivery.

A manner to overcome this limitation is to use a \emph{swarm} of drones. The idea of this technique is to position drones in recharging bases all over the delivery area. Therefore, a package can be delivered from one place to another through the collective effort of such drones, which can exchange packets among them to achieve a faster delivery. One may note that, when not carrying a package, a drone is stationed in its recharging base, waiting for the next package arrival. The problem of computing a package delivery schedule with minimum delivery time for a single package is called the \fastdel problem~\cite{Bartschi2018}.

\sloppy
We can model the input to the \fastdel problem as a graph $G = (V,E)$ with $|V| = n$ and $|E| = m$, with a positive length $l_e$ associated with each edge $e \in E$, and a set of $k$ \emph{autonomous mobile agents} (\eg, autonomous drones) located initially on distinct nodes $p_1,p_2,\ldots,p_k$ of $G$.
Each agent $i$ has a predefined
velocity $v_i > 0$. Mobile agent $i$ can traverse an edge $e$ of the graph in $l_e/v_i$ time.
The package handover between agents can be done on the nodes of the graph or in any point of the graph's edges, as exemplified in Fig.~\ref{fig:package-exchange}. The objective of \fastdel is to deliver a single package, initially located in a source node $s \in V$, to a target node $y \in V$ while minimizing
the delivery time~$\mathcal{T}$.

B\"artschi et al.~\cite{Bartschi2018} also consider the case where each agent $i$ is
additionally associated with a weight $\omega_i>0$ and consumes energy
$\omega_i\cdot l_e$ when traversing edge~$e$. For this model,
the total energy consumption $\mathcal{E}$ of a solution becomes relevant as well,
and one can consider the objective of minimizing $\mathcal{E}$ among all solutions that have the minimum
delivery time $\mathcal{T}$ (or vice versa), or of minimizing a convex combination $\varepsilon \cdot \mathcal{T} + (1 - \varepsilon) \cdot \mathcal{E}$ for a given $\varepsilon \in (0, 1)$.

\fussy

\begin{figure*}[t]
	\centering
	\subfloat[]{\label{subfig:exchange-node}
		\resizebox{0.65\textwidth}{!}{\input{figures/package-exchange-node.tex}}}

	\subfloat[]{\label{subfig:exchange-edge}
	    \resizebox{0.65\textwidth}{!}{\input{figures/package-exchange-edge.tex}}}

	\caption{\protect\subref{subfig:exchange-node} Package exchange on a node; \protect\subref{subfig:exchange-edge} package exchange on an edge.}
	\label{fig:package-exchange}
\end{figure*}

\subsection{Related Work} \label{subsec:related}

The problem of delivering packages through a swarm of autonomous drones has been studied in the literature. The work of Bärtschi et al.~\cite{Bartschi2017} considers the problem of delivering packages while minimizing the total energy consumption of the drones. In their work, all drones have the same velocity but may have different weights, and the package's exchanges between drones are restricted to take place on the graph's nodes. They show that this problem is NP-hard when an arbitrary number of packages need to be delivered, but can be solved in polynomial time for a single package, with complexity $\mathcal{O}(k + n^3)$. 

When minimizing only the delivery time $\mathcal{T}$, one can solve the problem of delivering a single package with autonomous mobile agents with different velocities in polynomial-time: B\"artschi et al.~\cite{Bartschi2018} gave an $\mathcal{O}(k^2 m + k n^2 + \APSP)$ algorithm for this problem, where $\APSP$ stands for the time complexity of the All-Pairs Shortest Paths problem.  

Some work in the literature considered the minimization of both the total delivery time and the energy consumption. It was shown that the problem of delivering a single package with autonomous agents of different velocities and weights is solvable in polynomial-time when lexicographically minimizing the tuple ($\mathcal{E},\mathcal{T}$)~\cite{Bartschi2017b}. On the other hand, it is NP-hard to lexicographically minimize the tuple ($\mathcal{T},\mathcal{E}$) or a convex combination of both parameters~\cite{Bartschi2018}.

A closely related problem is the \texttt{Budgeted Delivery Problem} (\texttt{BDP})~\cite{Chalopin2014,Chalopin2013,Bartschi2017c}, in which a package needs to be delivered by a set of energy-constrained autonomous mobile agents. In \texttt{BDP}, the objective is to compute a route to deliver a single package while respecting the energy constraints of the autonomous mobile agents. This problem in weakly NP-hard in line graphs~\cite{Chalopin2014} and strongly NP-hard in general graphs~\cite{Chalopin2013}. A variant of this problem is the \texttt{Returning Budgeted Delivery Problem} (\texttt{RBDP})~\cite{Bartschi2017c}, which has the additional constraint that the energy-constrained autonomous agents must return to their original bases after carrying the package. Surprisingly, this new restriction makes \texttt{RBDP} solvable in polynomial time in trees. However, it is still strongly NP-hard even for planar graphs.

\subsection{Our Contribution} \label{subsec:contribution}

This paper deals with the \fastdel problem. We focus on the first objective, \ie, computing delivery schedules with the minimum delivery time. We provide an $\mathcal{O}(kn\log n+km)$ time algorithm for \fastdel, which is more efficient than the previous $\mathcal{O}(k^2 m + k n^2 + \APSP)$ time algorithm for this problem~\cite{Bartschi2018}. 

Preliminaries are presented in Sect.~\ref{sec:prelim}.
We then describe our algorithm to solve \fastdel in Sect.~\ref{sec:fdp}. The algorithm uses as a subroutine, called once for each edge of $G$, an algorithm for a problem that we refer to as \fastldel, which is presented in Sect.~\ref{sec:linedelivery}.

\section{Preliminaries}
\label{sec:prelim}%
As mentioned earlier, in the \fastdel problem
we are given an undirected graph $G=(V,E)$ with $n=|V|$ nodes
and $m=|E|$ edges. Each edge $e\in E$ has a positive length $l_e$.
We assume that a path can start on a node or in some point in the
interior of an edge. Analogously, it can end on another node or in 
some point in the interior of an edge. The length of a path is equal 
to the sum of the lengths of its edges. If a path starts or ends at 
a point in the interior of an edge, only the portion of its length 
that is traversed by the path is counted. For example, a path that is
entirely contained in an edge $e=\{u,v\}$ of length $l_e=10$ and
starts at distance $2$ from $u$ and ends at distance $5$ from $u$ has
length~$3$.

We are also given $k$ mobile agents, which are initially located at nodes
$p_1,p_2,\ldots,p_k \in V$. Each agent $i$ has a positive velocity (or speed)
$v_i$, $1\le i\le k$. A single package is located initially
(at time~$0$) on a given source node $s\in V$ and needs to be
delivered to a given target node $y\in V$. An agent can pick up
the package in one location and drop it off (or hand it to another
agent) in another one. An agent with velocity $v_i$ takes 
time~$d/v_i$ to carry a package over a path of length~$d$.
The objective of \fastdel is to determine a schedule for the agents
to deliver the package to node~$y$ as quickly as possible, \ie,
to minimize the time when the package reaches~$y$.

We assume that there is at most one agent on each node. 
This assumption can be justified by the fact that, if there were several agents 
on the same node, we would use only the fastest one among them. Therefore,
as already observed in \cite{Bartschi2018}, after a preprocessing step
running in time $\mathcal{O}(k+|V|)$, we may assume that $k\le n$.

The following lemma from \cite{Bartschi2018} establishes some useful
properties of an optimal delivery schedule for the mobile agents.

\begin{lemma}[B\"artschi et al., 2018]
\label{lem:prop}%
For every instance of \fastdel, there is an optimum solution in which
$(i)$ the velocities of the involved agents are strictly increasing, and
$(ii)$ no involved agent arrives on its pick-up location earlier than the
package (carried by the preceding agent).
\end{lemma}

\section{Algorithm for the Fast Delivery Problem}
\label{sec:fdp}%
B\"artschi et al.~\cite{Bartschi2018} present a dynamic programming
algorithm that computes an optimum solution for \fastdel in
time $\mathcal{O}(k^2 m + kn^2 + \APSP)\subseteq \mathcal{O}(k^2 n^2 + n^3)$,
where $\APSP$ denotes the time complexity of an algorithm for
solving the all-pairs shortest path problem.
In this paper we design an improved algorithm with running time
$\mathcal{O}(km + kn\log n) \subseteq \mathcal{O}(n^3)$ by showing
that the problem can be solved by adapting the approach of
Dijkstra's algorithm for edges with time-dependent transit
times~\cite{CH/66,DW/09}.

For any edge $\{u,v\}$, we denote by $a_t(u,v)$ the earliest
time for the package to arrive at~$v$ if the package is at node~$u$
at time $t$ and needs to be carried over the
edge~$\{u,v\}$. We refer to the problem of computing
$a_t(u,v)$, for a given value of $t$ that represents
the earliest time when the package can reach $u$, as \fastldel.
In Sect.~\ref{sec:linedelivery}, we will show
that \fastldel can be solved in $\mathcal{O}(k)$ time after a preprocessing
step that spends $\mathcal{O}(k\log k)$ time per node. Our preprocessing
calls \textsc{PreprocessReceiver}$(v)$ once for each node
$v\in V\setminus\{s\}$ at the start of the algorithm. Then,
it calls \textsc{PreprocessSender}$(u,t)$ once for each
node $u \in V$, where $t$ is the earliest time when the
package can reach~$u$. Both preprocessing steps run
in $\mathcal{O}(k\log k)$ time per node. Once both preprocessing
steps have been carried out, a call to
\textsc{FastLineDelivery}$(u,v,t)$ computes $a_t(u,v)$ in $\mathcal{O}(k)$ time.

\begin{algorithm}[!tbh]
\SetAlgoLined
\KwData{graph $G=(V,E)$ with positive edge lengths $l_e$ and source node $s\in V$, target node $y\in V$; $k$ agents with velocity $v_i$ and initial location $p_i$ for $1\le i\le k$}
\KwResult{earliest arrival time $\dist(y)$ for package at destination}
\Begin{compute $d(p_i,v)$ for $1\le i\le k$ and all $v\in V$\;
construct list $A(v)$ of agents in order of increasing arrival times and velocities for each $v\in V$\;
\textsc{PreprocessReceiver}$(v)$ for all $v\in V\setminus\{s\}$\;
$\dist(s)\leftarrow t_s$\tcc*{time when first agent reaches $s$}
$\dist(v)\leftarrow \infty$ for all $v\in V\setminus\{s\}$\;
$\final(v)\leftarrow \mathrm{false}$ for all $v\in V$\;
insert $s$ into priority queue $Q$ with priority $\dist(s)$\;
\While{$Q$ not empty}{
	$u\leftarrow$ node with minimum $\dist$ value in $Q$\;
	delete $u$ from $Q$\;
	$\final(u)\leftarrow \mathrm{true}$\;
	\If{$u=y$}{break\;}
	$t\leftarrow\dist(u)$\tcc*{time when package reaches $u$}
	\textsc{PreprocessSender}$(u,t)$\;
	\ForAll{neighbors $v$ of $u$ with $\final(v)=\mathrm{false}$}{
		$a_t(u,v) \leftarrow \textsc{FastLineDelivery}(u,v,t)$\;
		\If{$a_t(u,v)<\dist(v)$}{
			$\dist(v)\leftarrow a_t(u,v)$\;
			\eIf{$v\in Q$}{
				decrease priority of $v$ to $\dist(v)$\;
			      }{
			      	insert $v$ into $Q$ with priority $\dist(v)$\;
				}
		}
	}
}
\Return{$\dist(y)$}\;}
\caption{Algorithm for \fastdel}
\label{algo:dijkstra}
\end{algorithm}

Algorithm~\ref{algo:dijkstra} shows the pseudo-code for our 
solution for \fastdel. Initially, we run Dijkstra's algorithm to solve the single-source shortest paths problem for each node where an agent is located initially (line 2). This takes time $\mathcal{O}(k(n\log n + m))$ if we use the implementation of Dijkstra's algorithm with Fibonacci heaps as priority
queue~\cite{FT/87} and yields the distance $d(p_i,v)$
(with respect to edge lengths~$l_e$) between any node $p_i$ where
an agent is located and any node $v\in V$.
From this we compute, for every node~$v$, the earliest
time when each mobile agent can arrive at that node: The earliest
possible arrival time of agent $i$ at node $v$ is $a_i(v) = d(p_i,v)/v_i$.
Then, we create a list of the arrival times of the $k$ agents on each node (line 3).
For each node, we sort the list of the $k$ agents by ascending arrival time in
$\mathcal{O}(k\log k)$ time, or $\mathcal{O}(nk\log k)$ in total for all nodes. 
We then discard from the list of each node all agents that arrive at the same time 
or after an agent that is strictly faster. If several agents with the same
velocity arrive at the same time, we keep one of them
arbitrarily. Let $A(v)$ denote the resulting list for node~$v$. 
Those lists will be used in the solution of the
\fastldel problem described in Sect.~\ref{sec:linedelivery}.

For each node~$v$, we maintain a value
$\dist(v)$ that represents the current upper bound on the
earliest time when the package can reach~$v$ (lines 5 and 6).
The algorithm maintains a priority queue containing nodes
that have a finite $\dist$ value, with the $\dist$ value
as the priority (line 8). In each step, a node $u$ with minimum $\dist$
value is removed from the priority queue (lines 10 and 11), and the node
becomes \emph{final} (line 12). Nodes that are not final are called
\emph{non-final}. The $\dist$ value of a final node
will not change any more and represents the earliest time
when the package can reach the node (line 16). After $u$ has been
removed from the priority queue, we compute for each
non-final neighbor $v$ of $u$ the time $a_t(u,v)$, where
$t=\dist(u)$, by solving the \fastldel problem (line 19).
If $v$ is already in $Q$, we
compare $a_t(u,v)$ with $\dist(v)$ and, if $a_t(u,v)<\dist(v)$,
update $\dist(v)$ to $\dist(v)=a_t(u,v)$ and adjust the
priority of $v$ in $Q$ accordingly (line 23).
On the other hand, if $v$ is not yet in~$Q$, we set $\dist(v)=a_t(u,v)$
and insert $v$ into~$Q$ (line 25). 

Let $t_s$ be the earliest time when an agent
reaches~$s$ (or $0$, if an agent is located at $s$ initially).
Let $i'$ be that agent. As the package must stay at $s$ from
time $0$ to time~$t_s$, we can assume that $i'$ brings the
package to $s$ at time $t_s$.
Therefore, we initially set $\dist(s)=t_s$ and insert $s$ into the
priority queue~$Q$ with priority~$t_s$. The algorithm terminates when $y$
becomes final (line 14) and returns the value $\dist(y)$, \ie, the earliest time when the package can reach~$y$. The schedule that delivers the package to $y$ by time
$\dist(y)$ can be constructed in the standard way, by storing for each node~$v$
the predecessor node $u$ such that $\dist(v)=a_{\dist(u)}(u,v)$ and the schedule
of the solution to \textsc{FastLineDelivery}$(u,v,\dist(u))$.

\begin{theorem}
Algorithm~\ref{algo:dijkstra}
runs in $\mathcal{O}(kn\log n +km)$ time and
computes an optimal solution to the \fastdel problem.
\end{theorem}

\begin{proof}
First, we note that it is easy to see that $a_t(u,v)\le a_{t'}(u,v)$
holds for $t'\ge t$ in our setting: If the package arrives at $u$ at
time $t$ and if we had $a_{t'}(u,v)<a_t(u,v)$ for some $t'>t$,
the package could simply wait at $u$ until time $t'$ and then
get transported to $v$ in the same way as if it had
reached $u$ at time $t'$. The package would reach $v$ at
time $a_{t'}(u,v)$, contradicting the assumption that
$a_{t'}(u,v)<a_t(u,v)$.
Thus, the network has the FIFO property (or non-overtaking
property), and it is known that the modified Dijkstra
algorithm is correct for such networks~\cite{DW/09}.

Furthermore, we can observe that concatenating the
solutions of \fastldel (which are computed by Algorithm~\ref{algo:linedel}
in Sect.~\ref{sec:linedelivery} and which are correct
by Theorem~\ref{th:linedel} in Sect.~\ref{sec:linedelivery})
over the edges of the
shortest path from $s$ to $y$ determined by
Algorithm~\ref{algo:dijkstra} indeed gives a feasible
solution to \fastdel: Assume that the package reaches
$u$ at time $t$ while being carried by agent $i$ and is
then transported from $u$ to $v$ over edge $\{u,v\}$,
reaching $v$ at time $a_t(u,v)$. The only agents
involved in transporting the package from $u$ to
$v$ in the solution returned by \textsc{FastLineDelivery}$(u,v,t)$
will have velocity at least $v_i$ because agent $i$
arrives at $u$ before time $t$, \ie, $a_i(u)\le t$, and
hence no slower agent would be used to transport
the package from $u$ to~$v$. These agents
have not been involved in transporting the package
from $s$ to $u$ by property~(i) of Lemma~\ref{lem:prop},
except for agent~$i$ who is indeed
available at node $u$ from time~$t$.

The running time of the algorithm consists of the following
components: Computing standard shortest paths with respect
to the edge lengths $l_e$ from the locations
of the agents to all other nodes takes $\mathcal{O}(k(n\log n+m))$ time.
The time complexity of the Dijkstra algorithm with
time-dependent transit times for a graph with $n$ nodes
and $m$ edges is $\mathcal{O}(n\log n+m)$. The only extra work
performed by our algorithm consists of $\mathcal{O}(k\log k)$ pre-processing
time for each node and $\mathcal{O}(k)$ time per edge for solving the
\fastldel problem, a total of $\mathcal{O}(nk\log k+mk)=
\mathcal{O}(kn\log n+km)$ time.
\end{proof}

\section{An Algorithm for Fast Line Delivery}
\label{sec:linedelivery}%
In this section we present the solution to
\fastldel that was used as a subroutine in the
previous section. 
We consider the setting of a single edge $e=\{u,v\}$ with
end nodes $u$ and $v$. The objective is to deliver the package
from node $u$ to node~$v$ over edge~$e$ as quickly as possible.
In our illustrations, we use the
convention that $v$ is drawn on the left and $u$ is drawn
on the right. We assume that the package reaches $u$ at time $t$
(where $t$ is the earliest possible time when the package can reach~$u$)
while being carried by agent~$i_0$.

As discussed in the previous section, let $A(v)=(a_1,a_2,\ldots,a_{\ell})$
be the list of agents arriving at node~$v$ in order of
increasing velocities and increasing arrival times.
For $1\le i\le \ell$, denote by $t_i$ the time when $a_i$
reaches $v$, and by $v_i$ the velocity of agent~$a_i$.
We have $t_i<t_{i+1}$ and $v_i<v_{i+1}$ for $1\le i<\ell$.

Let $B(u)=(b_1,b_2,\ldots,b_r)$ be the list of agents
with increasing velocities and increasing arrival times
arriving at node $u$, starting
with the agent $i_0$ whose arrival time is set to~$t$.
The list $B(u)$ can be computed from $A(u)$ in $\mathcal{O}(k)$ time
by discarding all agents slower than $i_0$ and setting
the arrival time of $i_0$ to~$t$. For $1\le i\le r$, let
$t_i'$ denote the time when $b_i$ reaches~$w$, and
let $v_i'$ denote the velocity of~$b_i$.
We have $t_i'<t_{i+1}'$ and $v_i'<v_{i+1}'$ for $1\le i<r$.

As $k$ is the total number of agents, we have
$\ell\le k$ and $r\le k$. In the following, we first
introduce a geometric representation of the agents and
their potential movements in transporting the package
from $u$ to $v$ (Sect.~\ref{sec:geom}) and then
present the algorithm for \fastldel (Sect.~\ref{sec:linealgo}).

\subsection{Geometric Representation and Preprocessing}
Figure~\ref{fig:geom} shows a geometric representation 
of how agents $a_1,\ldots,a_\ell$ move towards $u$
if they start to move from $v$ to $u$ immediately
after they arrive at~$v$. 
The vertical axis represents time, and the horizontal
axis represents the distance from $v$ (in the direction
towards~$u$ or, more generally, any neighbor of $v$).
The movement of each
agent $a_i$ can be represented by a line
with the line equation $y=t_i+{x}/{v_i}$ (\ie, the $y$ value
is the time when agent $a_i$ reaches the point at
distance $x$ from $v$).
After an agent is overtaken by a faster agent, the
slower agent is no longer useful for picking up the
package and returning it to $v$, so we can discard
the part of the line of the slower agent that lies to the
right of such an intersection point with the line of
a faster agent. After doing this for all agents (only
the fastest agent $a_\ell$ does not get overtaken and will
not have part of its line discarded), we obtain a
representation that we call the \emph{relevant arrangement} $\Psi$ of the
agents $a_1,\ldots,a_\ell$. In the relevant arrangement,
each agent $a_i$ is represented by a line segment
that starts at $(0,t_i)$, lies on the
line $y=t_i+{x}/{v_i}$, and ends at the first intersection
point between the line for $a_i$ and the line of a
faster agent $a_j$, $j>i$. For the fastest agent $a_\ell$,
there is no faster agent, and so the agent is represented
by a half-line. One can view the relevant arrangement
as representing the set of all points where an agent from
$A(v)$ travelling towards~$u$ could receive the package from a slower
agent travelling towards~$v$.

\begin{figure}
\centering
{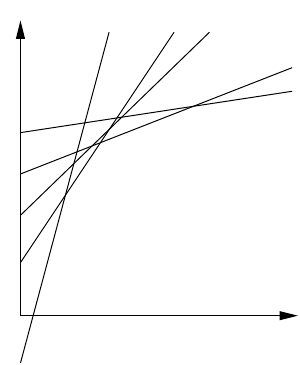}
\hspace{2cm}
{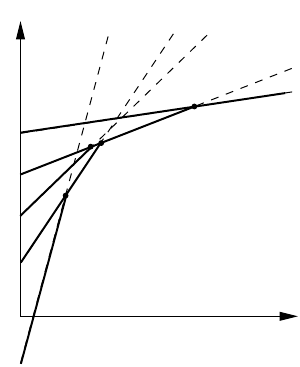}

\caption{Geometric representation of agents moving
from $v$ towards $u$ (left), and their relevant arrangement
with removed half-lines shown dashed (right)}
\label{fig:geom}
\end{figure}

The relevant arrangement has size $\mathcal{O}(k)$ because each intersection point
can be charged to the slower of the two agents that create the intersection.
It can be computed in $\mathcal{O}(k\log k)$
time using a sweep-line algorithm very similar to
the algorithm by Bentley and Ottmann~\cite{BO/79} for line segment
intersection. The relevant
arrangement is created by a call to \textsc{PreprocessReceiver}$(v)$ (see Algorithm~\ref{algo:preprocreceiver}).

\begin{algorithm}[htbp]
\SetAlgoLined
\KwData{Node $v$ (and list $A(v)$ of agents arriving at $v$)}
\KwResult{Relevant arrangement $\Psi$}
Create a line $y=t_i+{x}/{v_i}$ for each agent $a_i$ in $A(v)$\;
Use a sweep-line algorithm (starting at $x=0$, moving towards larger $x$ values) to construct the relevant arrangement $\Psi$\;
\caption{Algorithm \textsc{PreprocessReceiver}$(v)$}
\label{algo:preprocreceiver}
\end{algorithm}

\begin{algorithm}[htbp]
\SetAlgoLined
\KwData{Node $u$ (and list $A(u)$ of agents arriving at $u$), time $t$ when package arrives at~$u$ (carried by agent $i_0$)}
\KwResult{Lower envelope $L$ of agents carrying package away from $u$}
$B(u)\leftarrow A(u)$ with agents slower than $i_0$ removed and arrival time of~$i_0$ set to $t$\;
Create a line $y=t_i'-{x}/{v_i'}$ for each agent $b_i$ in $B(v)$\;
Use a sweep-line algorithm (starting at $x=0$, moving towards smaller $x$ values) to construct the lower envelope $L$;
\caption{Algorithm \textsc{PreprocessSender}$(u,t)$}
\label{algo:preprocsender}
\end{algorithm}

For the agents in the list $B(u)=(b_1,\ldots,b_r)$ that move
from $u$ towards $v$, we use a similar representation. However, in this case we only need to determine the lower envelope of the lines representing the
agents. See Fig.~\ref{fig:envelope} for an example.
The lower envelope $L$ can be computed in $\mathcal{O}(k\log k)$ time
(e.g., using a sweep-line algorithm, or via computing the
convex hull of the points that are dual to the lines~\cite[Sect.~11.4]{MOMM/08}).
The call \textsc{PreprocessSender}$(u,t)$ (see Algorithm~\ref{algo:preprocsender}) determines the
list $B(u)$ from $A(u)$ and $t$ in $\mathcal{O}(k)$ time and
then computes the lower envelope of the agents in $B(u)$
in time $\mathcal{O}(k\log k)$.
When we consider a particular edge $e=\{u,v\}$, we place
the lower envelope $L$ in such a way that the position on the $x$-axis
that represents $u$ is at $x=l_e$. We say in this
case that the lower envelope is \emph{anchored} at $x=l_e$.
Algorithm~\ref{algo:preprocsender} creates the lower envelope anchored at $x=0$,
and the lower envelope anchored at $x=l_e$ can be obtained by shifting it right
by~$l_e$.
\label{sec:geom}%
\begin{figure}
\centerline{
{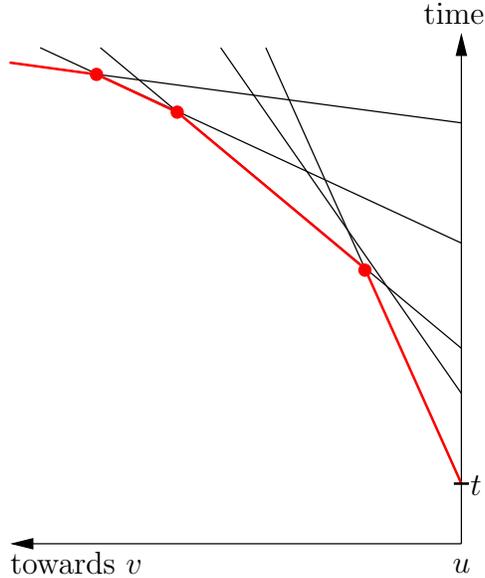}
}
\caption{Geometric representation of agents moving
from $u$ towards $v$ (lower envelope marked in red)}
\label{fig:envelope}
\end{figure}

\subsection{Main Algorithm}
\label{sec:linealgo}%
Assume we have computed the relevant arrangement $\Psi$
of the agents in the list $A(v)=(a_1,\ldots,a_\ell)$
and the lower envelope
$L$ of the lines representing the agents in the list
$B(u)=(b_1,b_2, \ldots,b_r)$.

The lower envelope $L$ of the agents in $B(u)$ represents
the fastest way for the package to be transported
from $u$ to $v$ if only agents in $B(u)$ contribute
to the transport and these agents move from $u$ towards
$v$ as quickly as possible. At each time point during
the transport, the package is at the closest point
to $v$ that it can reach if only agents in $B(u)$
travelling from $u$ to $v$ contribute to its transport.
We say that such a schedule where the package is as close
to $v$ as possible
at all times is \emph{fastest
and foremost} (with respect to a given set of agents).

The agents in $A(v)$ can potentially speed up the
delivery of the package to $v$ by travelling towards~$u$,
picking up the package from a slower agent that is currently
carrying it, and then turning around and
moving back towards $v$ as quickly as possible.
By considering intersections between $L$ and the
relevant arrangement $\Psi$ of $A(v)$, we can find
all such potential handover points. More precisely, we
trace $L$ from $u$ (\ie, $x=d(u,v)$) towards $v$ (\ie, $x=0$).
Assume that $q$ is the first point where a handover is possible.
If a faster agent $j$ from $A(v)$ can receive the package
from a slower agent $i$ at point $q$ of~$L$, we update
$L$ by computing the lower envelope of $L$ and
the half-line representing the agent $j$ travelling
from point~$q$ towards~$v$.
If the intersection point is with an agent $j$ from $A(v)$
that is not faster than the agent $i$ that is currently
carrying the package, we ignore the intersection point.
We then continue to trace $L$ towards $v$ and process
the next intersection point in the same way.
We repeat this step
until we reach~$v$ (\ie, $x=0$). The final $L$ represents
an optimum solution to the \fastldel problem,
and the $y$-value of $L$ at $x=0$ represents
the arrival time of the package at~$v$.
See Algorithm~\ref{algo:linedel} for
pseudo-code of the resulting algorithm.

\begin{algorithm}[htbp]
\SetAlgoLined
\KwData{Edge $e=\{u,v\}$, earliest arrival time $t$ of package at $u$,
  lists $A(u)$ and~$A(v)$}
\KwResult{Earliest time when package reaches $v$ over edge $\{u,v\}$}
\tcc{Assume \textsc{PreprocessReceiver}$(v)$ and \textsc{PreprocessSender}$(u,t)$
have already been called.}
$L\leftarrow$ lower envelope of agents $B(u)$ anchored at $x=l_e$\;
$\Psi\leftarrow$ relevant arrangement of $A(v)$\;
start tracing $L$ from $u$ (\ie, $x=l_e$) towards $v$ (\ie, $x=0$)\;
\While{$v$ (\ie, $x=0$) is not yet reached}{
  $q \leftarrow$ next intersection point of $L$ and $\Psi$\;
  \tcc{assume $q$ is intersection of agent $i$ from $L$ and agent $j$ from $\Psi$}
  \eIf{$v_j>v_i$}{
	replace $L$ by the lower envelope of $L$ and the line for agent $j$
	moving left from point $q$\;\label{mergelower}
  }{
  ignore $q$
  }
}
\Return{$y$-value of $L$ at $x=0$}
\caption{Algorithm \textsc{FastLineDelivery}$(u,v,t)$}
\label{algo:linedel}
\end{algorithm}

An illustration of step \ref{mergelower} of Algorithm~\ref{algo:linedel},
which updates $L$ by incorporating a faster agent from $A(v)$,
is shown in Fig.~\ref{fig:modL}. Note that the time for executing
this step is $\mathcal{O}(g)$, where $g$ is the number of segments removed from
$L$ in the operation. As a line segment corresponding to an agent
can only be removed once, the total time spent in executing
step \ref{mergelower} (over all executions of step \ref{mergelower}
while running Algorithm~\ref{algo:linedel}) is $\mathcal{O}(k)$.
\begin{figure}
\centerline{
\scalebox{0.9}{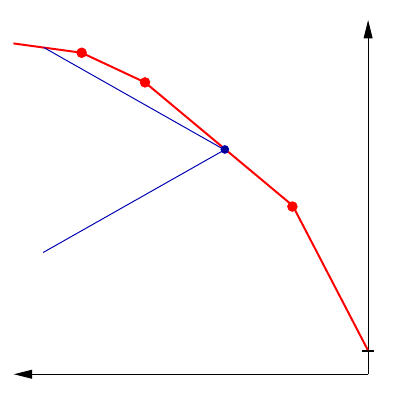}
\hspace{2em}
\scalebox{0.9}{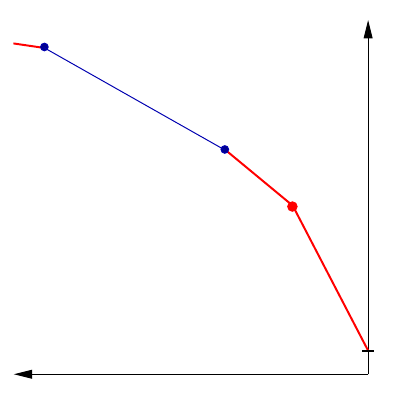}
}
\caption{Agent $i$ meets a faster agent $j$ at intersection
point $q$ (left). The part of $L$ from $q$ to $q'$ has been
replaced by a line segment representing agent~$j$ carrying the
package towards $v$ (right).}
\label{fig:modL}
\end{figure}

Finally, we need to analyze how much time is spent in finding
intersection points with line segments of the relevant arrangement $\Psi$
while following the lower envelope $L$ from $u$ to $v$.
See Fig.~\ref{fig:intersections} for an illustration.
We store the relevant arrangement using standard data structures
for planar arrangements~\cite{GR/00}, so that we can follow the edges of each
face in clockwise or counter-clockwise direction efficiently (\ie,
we can go from one edge to the next in constant time) and move
from an edge of a face to the instance of the same edge in the
adjacent face in constant time. This representation also allows
us to to trace the lower envelope of $\Psi$ in time $\mathcal{O}(k)$.

First, we remove from $\Psi$ all line segments corresponding
to agents that are not faster than $i_0$ (recall that $i_0$
is the agent that brings the package to node $u$ at time $t$).
Then, in order to find the first intersection point $q_1$ between
$L$ and $\Psi$, we can trace $L$ and the lower envelope of $\Psi$ from $u$ 
towards $v$ in parallel until they meet. One may observe that $L$ cannot 
be above the lower envelope of $\Psi$ at $u$ because otherwise an agent 
faster than $i_0$ reaches $u$ before time $t$, and that agent could pick 
up the package from $i_0$ before time $t$ and deliver it to $u$ before 
time $t$, a contradiction to $t$ being the earliest arrival time 
for the package at $u$. This takes $\mathcal{O}(k)$ time.
After computing one intersection point $q_i$ (and possibly updated
$L$ as shown in Fig.~\ref{fig:modL}), we find the next intersection
point by following the edges on the inside of the next face in
counter-clockwise direction until we hit $L$ again at $q_{i+1}$.
This process is illustrated by the blue arrow in Fig.~\ref{fig:intersections}, 
showing how $q_2$ is found starting from $q_1$. Hence, the total time
spent in finding intersection points is bounded by the initial size of $L$
and the number of edges of all the faces of the relevant arrangement, which is
$\mathcal{O}(k)$.

\begin{figure}[!htb]
\centerline{
{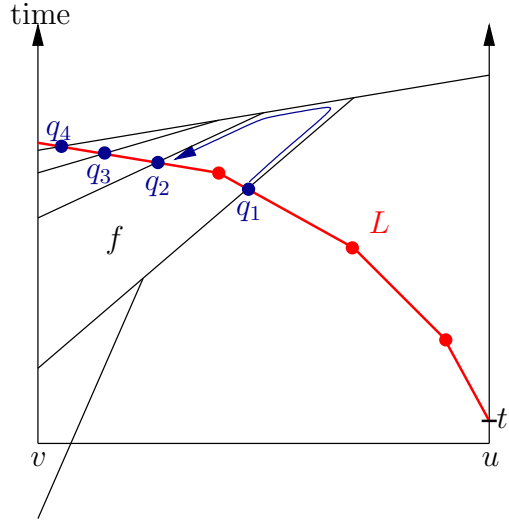}
}
\caption{Intersection points $q_1$, $q_2$, $q_3$, $q_4$ between
the lower envelope $L$ (red) and the relevant arrangement $\Psi$.
Point $q_2$ is found from $q_1$ by simultaneously tracing $L$ and
the edges of the face $f$ of $\Psi$ in counter-clockwise direction.}
\label{fig:intersections}
\end{figure}

\begin{theorem}
\label{th:linedel}%
Algorithm~\ref{algo:linedel} solves \textsc{FastLineDelivery}$(u,v,t)$
and runs in $\mathcal{O}(k)$ time, provided that
\textsc{PreprocessReceiver}$(v)$ and
\textsc{PreprocessSender}$(u,t)$, which take
time $\mathcal{O}(k\log k)$ each, have already been executed.
\end{theorem}

\begin{proof}
The claimed running time follows from the discussion above.
Correctness follows by observing that the following invariant holds:
If the algorithm has traced $L$ up to position $(x_0,y_0)$,
then the current $L$ represents the fastest and foremost solution for
transporting the package from $u$ to $v$ using only agents in $B(u)$
and agents from $A(v)$ that can reach the package by
time~$y_0$.
\end{proof}

\section*{Acknowledgments}
Iago A. Carvalho was financed in part by the Coordenação de Aperfeiçoamento de Pessoal de Nível Superior - Brasil (CAPES) - Finance Code 001.

\bibliographystyle{plainurl}
\bibliography{fastdelivery}

\end{document}